\documentclass[a4paper,11pt]{article}
\usepackage{a4wide}
\usepackage{amsmath}
\usepackage{amssymb}
\usepackage{graphicx}
\usepackage{xcolor}
\usepackage{todonotes}
\usepackage{hyperref}
\usepackage{xspace}
\usepackage[labelfont=bf,font=small]{caption}

\usepackage{amsthm}

\graphicspath{{figures/}}

\newcommand{\R}{\mathbb{R}}
\newcommand{\ER}{\ensuremath{\exists\R}\xspace}

\newtheorem{definition}{Definition}
\newtheorem{conjecture}{Conjecture}

\newtheorem{corollary}{Corollary}
\newtheorem{theorem}{Theorem}
\newtheorem{lemma}{Lemma}
\newtheorem{problem}{Problem}

\DeclareMathOperator{\mcr}{max-cr}
\DeclareMathOperator{\mwtcr}{max-wt-cr}
\DeclareMathOperator{\mrcr}{max-\overline{cr}}
\DeclareMathOperator{\mccr}{max-cr^\circ}
\DeclareMathOperator{\mcut}{max-cut}
\DeclareMathOperator{\bcr}{bcr}
\DeclareMathOperator{\obf}{obf}

\title{On the Maximum Crossing Number\thanks{A preliminary version of
    this paper appeared in the Proceedings of the 28th International
    Workshop on Combinatorial Algorithms (IWOCA 2017).}}

\author{Markus~Chimani\thanks{Universit\"at Osnabr\"uck,
    \texttt{markus.chimani@uni-osnabrueck.de}} \and
  Stefan~Felsner\thanks{Technische Universit\"at Berlin,
    \texttt{felsner@math.tu-berlin.de}} \and
  Stephen~Kobourov\thanks{University of Arizona,
    \texttt{kobourov@cs.arizona.edu}} \and
  Torsten~Ueckerdt\thanks{Karlsruhe Institute of Technology,
    \texttt{torsten.ueckerdt@kit.edu}} \and Pavel~Valtr\thanks{Charles
    University, \texttt{valtr@kam.mff.cuni.cz}} \and
  Alexander~Wolff\thanks{Universit\"at W\"urzburg,
    \href{http://orcid.org/0000-0001-5872-718X}{orcid.org/0000-0001-5872-718X}}}

\date{}

\begin{document}

\maketitle

\begin{abstract}
  Research about crossings is typically about minimization.  In this
  paper, we consider \emph{maximizing} the number of crossings over
  all possible ways to draw a given graph in the plane.  Alpert et
  al.\ [Electron.~J.\ Combin., 2009] conjectured that any graph has a
  \emph{convex} straight-line drawing, e.g., a drawing with vertices
  in convex position, that maximizes the number of edge crossings.  We
  disprove this conjecture by constructing a planar graph on twelve
  vertices that allows a non-convex drawing with more crossings than
  any convex one.  Bald et al.\ [Proc.\ COCOON, 2016] showed that it
  is NP-hard to compute the maximum number of crossings of a geometric
  graph and that the weighted geometric case is NP-hard to
  approximate.  We strengthen these results by showing hardness of
  approximation even for the unweighted geometric case and prove that
  the unweighted topological case is NP-hard.
\end{abstract}

\section{Introduction}

While traditionally in graph drawing one wants to minimize the number
of edge crossings, we are interested in the opposite problem.
Specifically, given a graph $G$, what is the maximum number of edge
crossings possible, and what do embeddings\footnote{We consider only
  embeddings where vertices are mapped to distinct points in the plane
  and edges are mapped to continuous curves containing no vertex
  points other than those of their end vertices.} of $G$ that attain
this maximum look like?  Such questions have first been asked as early
as in the 19th
century~\cite{baltzer1885erinnerung,staudacher1893lehrbuch}.  Perhaps
due to the counterintuitive nature of the problem (as illustrated by
the disproved conjecture below) and due to the lack of established
tools and concepts, little is known about maximizing the number of
crossings.

Besides the theoretical appeal of the problem, motivation for this
problem can be found 
in analyzing the worst-case scenario when edge crossings are
undesirable but the placement of vertices and edges cannot be
controlled.

There are three natural variants of the crossing maximization problem
in the plane.  In the \emph{topological} setting, edges can be drawn
as curves, so that any pair of edges crosses at most once, and
incident edges do not cross.  In the straight-line variant (known for
historical reasons as the \emph{rectilinear} setting), edges must be
drawn as straight-line segments.  If we insist that the vertices are
placed in convex position (e.g., on the boundary of a disk or a convex
polygon) and the edges must be routed in the interior of their convex
hull, the topological and rectilinear settings are equivalent,
inducing the same number of crossings: the number only depends on the
order of the vertices along the boundary of the disk.  In this
\emph{convex setting}, a pair of edges crosses if and only if its
endpoints alternate along the boundary of the convex hull.

\paragraph{The topological setting.}

The maximum crossing number was introduced by
Ring\-el~\cite{ringel1963extremal} in 1963 and independently by
Gr\"unbaum~\cite{grunbaum1972arrangements} in 1972.

\begin{definition}[\cite{survey}]
  The \emph{maximum crossing number} of a graph $G$, $\mcr(G)$, is the
  largest number of crossings in any topological drawing of $G$ in
  which no three distinct edges cross in one point and every pair of
  edges has at most one point in common (a shared endpoint counts,
  touching points are forbidden).
\end{definition}

In particular, $\mcr(G)$ is the maximum number of crossings in the
topological setting.  Note that only independent pairs of edges, that
is those edge pairs with no common endpoint, can cross.  The number of
independent pairs of edges in a graph $G = (V,E)$ is given by $M(G) :=
{|E|\choose 2} - \sum_{v \in V} {\deg(v)\choose 2}$, a parameter
introduced by Piazza~\textit{et~al.}~\cite{piazza1991properties}.  For
every graph $G$, we have $\mcr(G) \leq M(G)$, and graphs for which
equality holds are known as \emph{thrackles} or
\emph{thrackable}~\cite{woodall71}.  Conway's Thrackle
Conjecture~\cite{lovasz1997conway} states that thrackles are precisely
the pseudoforests (graphs in which every connected component has at
most one cycle) in which there is no cycle of length four and at most
one odd cycle.  Equivalently, this famous conjecture states that
$\mcr(G) = M(G)$ implies $|E(G)| \leq |V(G)|$~\cite{woodall71}.

Another famous open problem is the Subgraph Problem posed by Ringeisen
et al.~\cite{ringeisen1991subgraphs}: Is it true that whenever $H$ is
a subgraph or induced subgraph of $G$, then we have $\mcr(H) \leq
\mcr(G)$?

Let us remark that allowing pairs of edges to only touch without
properly crossing each other, would indeed change the problem.  For
example, the $4$-cycle $C_4$ has two pairs of independent edges, and
$C_4$ can be drawn with one pair crossing and the other pair touching,
but $C_4$ is not thrackable; it is impossible to draw $C_4$ with both
pairs crossing, i.e., $\mcr(C_4)$ is $1$ and not $2$.

It is known that $\mcr(K_n)=\binom{n}{4}$~\cite{ringel1963extremal}
and that every tree is thrackable, i.e., $\mcr(G) = M(G)$ whenever $G$
is a tree~\cite{piazza1991properties}.  We refer to Schaefer's
survey~\cite{survey} for further known results on the maximum crossing
numbers of several graph classes.

\paragraph{The straight-line setting.}

The maximum rectilinear crossing number was introduced by
Gr\"un\-baum~\cite{grunbaum1972arrangements}; see
also~\cite{furry1977maximal}.

\begin{definition} 
  The \emph{maximum rectilinear crossing number} of a graph $G$,
  $\mrcr(G)$, is the largest number of crossings in any straight-line
  drawing of $G$.
\end{definition}

For every graph $G$, we have $\mrcr(G) \leq \mcr(G) \leq M(G)$, where
each inequality is strict for some graphs, while equality is possible
for other graphs.  For example, for the $n$-cycle $C_n$ we have
$\mrcr(C_n) = \mcr(C_n) = M(C_n) = n(n-3)/2$ for
odd~$n$~\cite{woodall71}, while $\mrcr(C_n) = M(C_n) - n/2+1$ and
$\mcr(C_n) = M(C_n)$ for even~$n$ different than
four~\cite{steinitz1923maximalzahl,alpert2009}.  For further
rectilinear crossing numbers of specific graphs we again refer to
Schaefer's survey~\cite{survey}.


For several graph classes, such as trees, the maximum (topological)
crossing number $\mcr(G)$ is known exactly, while little is known
about the rectilinear crossing number $\mrcr(G)$.  For planar graphs,
Verbitsky~\cite{verbitsky2008obfuscation} studied what he called the
\emph{obfuscation number}.  He defined $\obf(G)=\mrcr(G)$ and showed
that $\obf(G)<3|V(G)|^2$.  Note that this holds only for planar
graphs. 
For maximally planar graphs, that is, triangulations, Kang et
al.~\cite{kprsv08} give a ($56/39-\varepsilon$)-approximation for
computing $\mrcr(G)$.

\paragraph{The convex setting.}

It is easy to see that in the convex setting we may assume, without
loss of generality, that all vertices are placed on a circle and edges
are drawn as straight-line segments.  In fact, if the vertices are in
convex position and edges are routed in the interior of the convex
hull of all vertices, then a pair of edges is crossing if and only if
the vertices of the two edges alternate in the circular order along
the convex hull.

\begin{definition} 
  The {\em maximum convex crossing number} of a graph $G$, $\mccr(G)$,
  is the largest number of crossings in any drawing of~$G$ where the
  vertices lie on the boundary of a disk and the edges in the
  interior.
\end{definition}

From the definitions we now have that, for every graph $G$,
\begin{equation}
  \mccr(G) \leq \mrcr(G) \leq \mcr(G) \leq M(G),\label{eq:four-parameters}
\end{equation}
but this time it is not clear whether or not the first inequality can
be strict.  It is tempting (and rather intuitive) to say that in order
to get many crossings in the rectilinear setting, all vertices should
always be placed in convex position.  In other words, this would mean
that the maximum rectilinear crossing number and maximum convex
crossing number always coincide.  Indeed, this has been conjectured by
Alpert et al.\ in 2009.

\begin{conjecture}[Alpert et al.~\cite{alpert2009}] \label{c-general}
  Any graph~$G$ has a drawing with vertices in convex position that
  has $\mrcr(G)$ crossings, that is, $\mrcr(G)=\mccr(G)$.
\end{conjecture}

\paragraph{Our contribution.}
Our main result is that Conjecture~\ref{c-general} is false.  We
provide several counterexamples in Section~\ref{sec:counterex}.  There
we first present a rather simple analysis for a counterexample with
$37$ vertices. We then improve upon this by showing that the planar
$12$-vertex graphs shown in the middle of
Figure~\ref{fig:small-counterexamples} are counterexamples as well.
Before we get there, we discuss the four parameters
in~\eqref{eq:four-parameters} and relations between them in more
detail, and introduce some new problems in
Section~\ref{sec:preliminaries}.  Finally, in
Section~\ref{sec:complexity}, we investigate the complexity and
approximability of crossing maximization and show that the topological
problem is NP-hard, while the rectilinear problem is even hard to
approximate.

\section{Preliminaries and Basic Observations}
\label{sec:preliminaries}

Here we discuss the chain of inequalities
in~\eqref{eq:four-parameters} and extend it by several items.  Recall
that for a graph $G$, $M(G)$ denotes the number of independent pairs
of edges in~$G$.  By~\eqref{eq:four-parameters} we have $\mccr(G) \leq
M(G)$.  
We next show that this inequality is tight up to a factor of~$3$.  The
first part of the next lemma is due to
Verbitsky~\cite{verbitsky2008obfuscation}.

\begin{lemma}\label{lem:lowerbound-M}
  For every graph $G$, we have $M(G)/3 \leq \mccr(G)$.  Moreover, if
  $G$ has chromatic number at most $3$, then $M(G)/2 \leq \mccr(G)$.
\end{lemma}
\begin{proof}
  First, let $G$ be any graph.  We place the vertices of $G$ on a
  circle in a circular order chosen uniformly at random from the set
  of all their circular orders.  Then each pair of independent edges
  of $G$ is crossing with probability $1/3$ and there must be an
  ordering witnessing $\mccr(G)\geq M(G)/3$.
 
  Second, assume that $G$ can be properly colored with at most three
  colors.  In this case we place the vertices of $G$ on a circle in
  such a way that the three color classes occupy three pairwise
  disjoint arcs.  In each color class, we order the vertices randomly,
  choosing each linear order with the same probability.  Doing this
  independently for each color class, each pair of independent edges
  is crossing with probability $1/2$.  Hence, there must be an
  ordering witnessing $\mccr(G)\geq M(G)/2$.
\end{proof}

By Lemma~\ref{lem:lowerbound-M} we can extend the chain of
inequalities in~\eqref{eq:four-parameters} as follows: For every
graph~$G$, we have
\begin{equation}
  M(G) /3 \leq \mccr(G) \leq \mrcr(G) \leq \mcr(G) \leq M(G).\label{eq:four-parameters-new}
\end{equation}
%
The constant $1/3$ in the first inequality
in~\eqref{eq:four-parameters-new} cannot be improved: Consider the six
edges connecting a $4$-tuple of vertices in a rectilinear drawing of
the complete graph $K_n$. There is exactly one crossing among them if
the four vertices are in convex position, and there is no crossing
among them otherwise. It follows that the rectilinear maximum crossing
number of $K_n$ is attained if and only if the vertices are in convex
position, and in this case there are $M(K_n)/3=\binom{n}{4}$
crossings. Since Ringel~\cite{ringel1963extremal} proved $\mcr(K_n) =
\binom{n}{4}$, we get $\mccr(K_n) = \mrcr(K_n) = \mcr(K_n) = M(K_n)/3
= \binom{n}{4}$.

\medskip

We now introduce another item in the chain of
inequalities~\eqref{eq:four-parameters-new}.  We say that a
rectilinear drawing of a graph $G$ is \emph{separated} if there is a
line $\ell$ that intersects every edge of $G$.  Clearly, this is only
possible if $G$ is bipartite and in this case the line $\ell$
separates the vertices of the two color classes of $G$.

\begin{figure}[tb]
  \centering
  \includegraphics{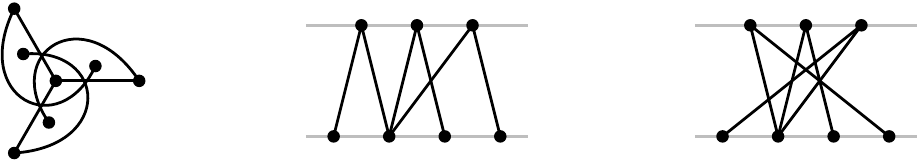}
  \caption{The smallest tree $G$ that is not a caterpillar with a
    topological drawing with $\mcr(G) = M(G) = 9$ crossings (left), a
    $2$-layer drawing with $\bcr(G) = 1$ crossings (middle) and a
    $2$-layer drawing with $\mrcr(G) = M(G) - \bcr(G) = 8$ crossings
    (right).}
  \label{fig:spider}
\end{figure}

Particularly nice are \emph{separated convex drawings}, i.e.,
separated drawings with vertices in convex position; see
Fig.~\ref{fig:spider} for an example.  Drawing bipartite graphs in the
separated convex model is equivalent to the $2$-layer model where the
vertices of the two color classes are required to be placed on two
parallel lines.  In this $2$-layer model, the crossing
\emph{minimization} of a bipartite graph $G$ has been studied under
the name \emph{bipartite crossing number}, denoted $\bcr(G)$.

\begin{lemma}\label{lem:max-to-min}
  For every bipartite graph $G$, the maximum number of crossings among
  all separated convex drawings of $G$ is exactly $M(G) - \bcr(G)$.
\end{lemma}
\begin{proof}
  Consider any separated convex drawing of any bipartite graph $G$.  A
  pair of independent edges is crossing if and only if their endpoints
  alternate along the convex hull.  So if $e_1 = u_1v_1$ and $e_2 =
  u_2v_2$ with $u_1,u_2$ being above the separating line $\ell$ and
  $v_1,v_2$ below, then $e_1$ and $e_2$ are crossing if in the
  circular order we see $u_1-u_2-v_1-v_2$, and non-crossing if we see
  $u_1-u_2-v_2-v_1$.  In particular, reversing the order of all
  vertices below the separating line $\ell$ transforms crossings into
  non-crossings and vice versa.  This shows that for a separated
  convex drawing with $k$ crossings, reversing results in exactly
  $M(G) - k$ crossings, which concludes the proof.
\end{proof}


Applying Lemma~\ref{lem:max-to-min} to the chain of
inequalities~\eqref{eq:four-parameters-new} shows that for every
bipartite graph~$G$ we have
\begin{equation}
  M(G)/2 \leq M(G) - \bcr(G) \leq \mccr(G) \leq \mrcr(G) \leq \mcr(G) \leq M(G).\label{eq:four-parameters-bipartite}
\end{equation}

It remains open whether the new inequality $M(G) - \bcr(G) \leq
\mccr(G)$ in~\eqref{eq:four-parameters-bipartite} is attained with
equality for every bipartite graph $G$.  For example, for a tree $G$
it is known, see e.g.~\cite{woodall71}, that $\mcr(G)=M(G)$, but it is
not hard to see that $\mrcr(G) = M(G)$ if and only if $G$ is a
caterpillar\footnote{A caterpillar is a tree in which all non-leaf
  vertices lie on a common path.}.  (Hence $\mrcr(G) < \mcr(G)$ holds
for every tree which is not a caterpillar.)  Moreover, it is equally
easy to see that a tree $G$ has a crossing-free $2$-layer drawing if
and only if $G$ is a caterpillar.  Thus, for every tree~$G$, we have
that $M(G) - \bcr(G) = M(G)$ if and only if $\mrcr(G) = M(G)$.  We
again refer to Fig.~\ref{fig:spider} for an illustration.

The expanded chain on
inequalities~\eqref{eq:four-parameters-bipartite}, leads to two
natural
questions:

\begin{problem}\label{c-bip1} \sloppy
  Does every bipartite graph $G$ have a separated drawing with
  $\mrcr(G)$ many crossings?  Does every tree $G$ have a separated
  convex drawing with $\mrcr(G)$ crossings, i.e., is $\mrcr(G) =
  M(G)-\bcr(G)$?
\end{problem}

Let us mention that Garey and Johnson~\cite{garey83} have shown that
bipartite crossing minimization is NP-hard.  The problem remains
NP-hard if the ordering of the vertices on one side is
prescribed~\cite{eades94}.  On trees, bipartite crossing minimization
can be solved efficiently~\cite{sssav01}.  For the one-sided two-layer
crossing minimization, Nagamochi~\cite{n-ib1sm-DCG05} gave an
$1.47$-approximation algorithm, improving upon the well-known median
heuristic, which yields a $3$-approximation~\cite{eades94}.  The
weighted case, which we define formally in
Section~\ref{sec:complexity}, admits a 3-approximation algorithm
\cite{ceks-cmwbg-JDA09}.

\section{Counterexamples for Conjecture~\ref{c-general}}
\label{sec:counterex}

In this section we present counterexamples for the convexity
conjecture. After some preliminary work we provide a counterexample
$H(4)$ on $37$ vertices.  To show that this graph is a counterexample,
we need to analyze only two cases.  (To show that $H(2)$ with $19$
vertices also is a counterexample would require more work.  Instead,
in Appendix~\ref{sec:small}, we prove that a certain planar subgraph
of $H(2)$ with only $12$ vertices and $16$ edges is already a
counterexample.)

A set of vertices $X\subset V$ in a graph $G=(V,E)$ is a set of
\emph{twins} if all vertices of $X$ have the same neighborhood in $G$
(in particular $X$ is an independent set). A \emph{vertex split} of
vertex $v$ in $G$ consists in adding a new vertex $v'$ to $G$ such
that $v'$ is a twin of $v$, that is, for any edge $vu$, there is an
edge $v'u$, and these are all the edges at~$v'$.

\begin{lemma}\label{lem:twins} 
  For any graph $G$ there is a convex drawing of $G$ maximizing the
  number of crossings among all convex drawings of $G$, such that each
  set of twins forms an interval of consecutive vertices along the
  convex hull of the drawing.
\end{lemma}

\begin{proof}
  Suppose $V_1,\dots,V_s$ are the maximal sets of twins in~$G$.
  Consider a convex drawing of~$G$ maximizing the number of
  crossings. It clearly suffices to show that for any set~$V_i$ we may
  move all the points of~$V_i$ next to one of the points of~$V_i$
  without decreasing the number of crossings, since this procedure
  done iteratively $s$ times, once for each of the sets
  $V_1,\dots,V_s$, results in a desired convex drawing of~$G$.

  We call a crossing \emph{$k$-rich} if there are $k$ vertices of
  $V_i$ among the four vertices of the edges forming the
  crossing. Since $V_i$ is independent, $k$ is $0$, $1$ or $2$ for
  each crossing. If we move only vertices of $V_i$ then $0$-rich
  crossings remain in the drawing.  If the vertices of $V_i$ appear in
  consecutive order along the convex hull of the drawing then the
  number of $2$-rich crossings is maximized due to the following
  argument. For any two vertices $u,v$ of $V_i$ and for any two
  neighbors $x,y$ of $V_i$, the $4$-cycle $uxvy$ is self-crossing
  which gives rise to a $2$-rich crossing. Since every $2$-rich
  crossing appears in a single $4$-cycle and every $4$-cycle can give
  rise to at most one crossing, the number of $2$-rich crossings is
  indeed maximized whenever the vertices of $V_i$ appear in
  consecutive order along the convex hull. It remains to show that
  there is a vertex $v$ in $V_i$ such that we can move the other
  vertices next to $v$ without decreasing the number of $1$-rich
  crossings. Each $1$-rich crossing involves exactly one vertex of
  $V_i$. The number of $1$-rich crossings involving a given vertex of
  $V_i$ is affected only by the position of that vertex and of the
  vertices of $V\setminus V_i$. Thus, if we choose $v$ as the vertex
  involved in the largest number of $1$-rich crossings and move all
  the other vertices of $V_i$ next to $v$, every vertex will be
  involved in at least as many $1$-rich crossings as it was before the
  vertices were moved.
\end{proof}

\paragraph{The construction of $H(k)$.}

For the construction of our example graphs $H(k)$, we start with a
9-cycle on vertices $v_0,\ldots,v_8$ with edges $v_i,v_{i+1}$ where
$i+1$ is to be taken modulo 9.  Add a `central' vertex $z$ adjacent to
$v_0,v_3,v_6$. This graph on 10 vertices is the base graph $H$.  The
example graph $H(k)$ is obtained from $H$ by applying $k$ vertex
splits to each of the nine cycle vertices~$v_i$. The graph $H(k)$ thus
consists of nine independent sets $V_i$ of size $k$ and the central
vertex~$z$.  In total it has $9k+1$ vertices and $9k^2+3k$ edges.
Figure~\ref{fig:example} (left) shows a schematic drawing of $H(k)$,
where each black edge represents a ``bundle'' of $k^2$ edges of $H(k)$
and each gray edge represents a set of $k$ edges.  We will show that
for $k\geq 4$ the drawing in Fig.~\ref{fig:example} (right) has more
crossings than any drawing with vertices in convex position.

\begin{figure}[tb]
  \centering
  \includegraphics{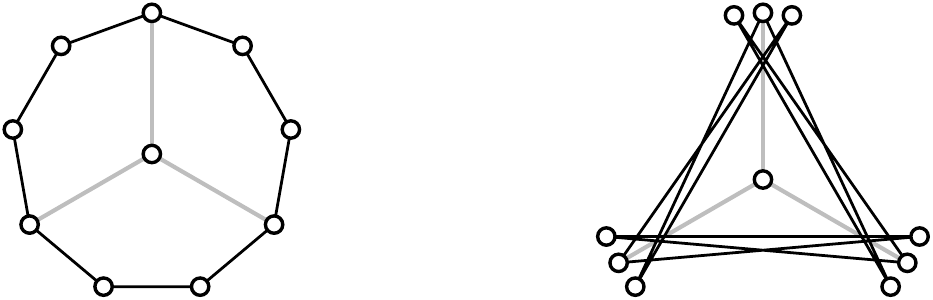}
  \caption{Left: The graph $H(k)$. Each circle represents $k$
    independent vertices, each black line segment represents a bundle
    of $k^2$ edges, each gray line segment represents $k$
    edges. Right: A non-convex drawing of $H(k)$.}
  \label{fig:example}
\end{figure}


From Lemma~\ref{lem:twins} we know that, in convex drawings of $H(k)$
with many crossings, the twin pairs of vertices can be assumed to be
next to each other.  Drawings of $H(k)$ of this kind are essentially
determined by the corresponding drawings of $H$, in which each set of
twins is represented just by one representative; see
Fig.~\ref{fig:example}. This justifies that later on we only look at
convex drawings of~$H$ with weighted crossings, and not of the full
$H(k)$.

An independent set of edges of $H(k)$ is \emph{weak} if the
corresponding edges in the base graph~$H$ are not independent; it is
\emph{strong} otherwise. The next lemma shows that our drawing
of~$H(k)$ realizes as many crossings on weak pairs of independent
edges as possible. This allows us to focus on strong pairs in the
subsequent analysis.

\begin{lemma}
  The drawing of $H(k)$ on the right side of Fig.~\ref{fig:example}
  maximizes the number of crossings on weak pairs of independent
  edges.
\end{lemma}

\begin{proof}
  Each edge $v_i,v_{i+1}$ of $H$ maps to a $K_{k,k}$ in $H(k)$. In the
  given drawing the $K_{k,k}$ is represented by a red edge. Since $V_i
  \cup V_{i+1}$ are in separated convex position the $K_{k,k}$
  contributes $\binom{k}{2}^2$ crossing.

  A pair of adjacent edges $v_{i-1},v_{i}$ and $v_i,v_{i+1}$ in $H$
  maps to a $K_{k,2k}$ in $H(k)$. We know that $\mrcr(K_{k,2k}) =
  \binom{k}{2}\binom{2k}{2}$ and this number of crossings is realized
  with separated convex position. In the drawing $V_i$ and
  $V_{i-1}\cup V_{i+1}$ are in separated convex position.

  A pair of adjacent edges $v_{i},z$ and $v_i,v_{i+1}$ in $H$ maps to
  a $K_{k,k+1}$ in $H(k)$.  Now we have $\mrcr(K_{k,k+1}) =
  \binom{k}{2}\binom{k+1}{2}$, and this number of crossings is realized
  with separated convex position of the vertices. In the drawing
  $V_i,V_{i+1}\cup \{z\}$ are in separated convex position. The case
  of adjacent edges $v_{i},z$ and $v_{i-1},v_{i}$ is identical.
\end{proof}

The remaining crossings of the drawing of $H(k)$ correspond to
crossings of two independent edges of $H$. These are either two red
edges or a red and a green edge of $H$. Red edges represent a bundle
of $k^2$ edges of $H(k)$ and green edges a bundle of $k$ edges of
$H(k)$. Hence a crossing of two red edges represents $k^4$ individual
crossing pairs and a crossing of a red and a green edge represent
$k^3$ individual crossing pairs.  We devide by $k^3$ and speak about a
crossing of two red edges as a crossing of weight $k$ and of a red
green crossing as a crossing of weight 1.  In the given drawing of
$H(k)$ every pair of red edges is crossing but every red edge has a
unique independent green edge which is not crossed. Hence, the weight
of the independent not crossing pairs of edges of~$H$ is~9.  We
summarize by saying that the given drawing has a weighted loss of~9.

\paragraph{The loss of convex drawings.}

We now study the weighted loss of convex drawings of $H$.  In a convex
drawing every red edge splits the 7 non-incident cycle vertices into
those on one side and those on the other side. The \emph{span} of a
red edge is the number of vertices on the smaller side.  Hence, the
span of an edge is one of 0, 1, 2, 3.
 
Let us consider the case where the 9-cycle is drawn with zero loss,
i.e., each red edge has span~3 and contributes a crossing with 6 other
red edges. The cyclic order of the cycle vertices is
$v_0,v_2,v_4,v_6,v_8,v_1,v_3,v_5,v_7$. Any two neighbors of $z$ have
the same distance in this cyclic order. Therefore, we may assume that
$z$ is in the short interval spanned by $v_0$ and $v_6$.  Every edge
of the 9-cycle is disjoint from at least one of the two green edges
$z,v_0$ and $z,v_6$ and the edge $v_7,v_8$ is disjoint from both. This
shows that the weighted loss of this drawing is at least~10.

A sequence of eight consecutive edges of span~3 forces the last edge
to also have span 3. Hence, we have at least two red edges~$e$ and~$f$
of span at most~2.  Each of these edges is disjoint from at least two
independent red edges. Since the two edges may be disjoint they
contribute a weighted loss of at least~$3k$. For $k>4$ this exceeds
the weighted loss of the drawing of Fig.~\ref{fig:example}.

\section{Complexity}
\label{sec:complexity}

Very recently, Bald et al.~\cite{Bald16} showed, by reduction from
\textsc{MaxCut}, that it is NP-hard to compute the maximum rectilinear
crossing number $\mrcr(G)$ of a given graph~$G$.  Their reduction also
shows that it is hard to approximate the weighted case better than
$\approx0.878$ assuming the Unique Games Conjecture and better than
$16/17$ assuming ${\cal P} \ne {\cal NP}$.  In the convex case, one
can ``guess'' the permutation; hence, this special case is in $\cal
NP$.  Bald et al.\ also stated that rectlinear crossing maximization
is similar to rectilinear crossing minimization in the sense that the
former ``inherits'' the membership in the class of the existential
theory of the reals (\ER), and hence in PSPACE, from the latter.  They
also showed how to derandomize Verbitsky's approximation
algorithm~\cite{verbitsky2008obfuscation} for $\mrcr$, turning the
expected approximation ratio of~$1/3$ into a deterministic one.

We now tighten the hardness results of Bald et al.\ by showing
APX-hardness for the \emph{unweighted} case.  Recall that
\textsc{MaxCut} is NP-hard to approximate beyond a factor of
$16/17$~\cite{h-soir-JACM01}.  Under the Unique Games Conjecture,
\textsc{MaxCut} is hard to approximate even beyond a factor of
$\approx\!0.878$~\cite{kkmo-oirm-SICOMP07}---the approximation ratio of
the famous semidefinite programming approach of Goemans and Williamson
\cite{gw-iaamc-JACM95} for \textsc{MaxCut}.  For a graph~$G$, let
$\mcut(G)$ be the maximum number of edges crossing a cut, over all
cuts of~$G$.

\begin{theorem}
  \label{thm:unweighted-rectilinear}
  Given a graph~$G$, $\mrcr(G)$ cannot be approximated better than
  \textsc{MaxCut}.
\end{theorem}

\begin{proof}
  As Bald et al., we reduce from \textsc{MaxCut}.  In their reduction,
  they add a large-enough set~$I$ of independent edges to the given
  graph~$G$.  They argue that $\mrcr(G+I)$ is maximized if the edges
  in~$I$ behave like a single edge with high weight that is crossed by
  as many edges of~$G$ as possible.  Indeed, suppose for a
  contradiction that, in a drawing with the maximum number of
  crossings, an edge~$e\in I$ crosses fewer edges than another
  edge~$e'$ in~$I$.  Then $e$ can be drawn such that its endpoints are
  so close to the endpoints of~$e'$ that both edges cross the same
  edges---and each other.  This would increase the number of
  crossings; a contradiction.  W.l.o.g., we can make the ``heavy
  edge'' so long that its endpoints lie on the convex hull of the
  drawing.  This means that the heavy edge induces a cut of~$G$.  The
  cut is maximum since the heavy edge can be made arbitrarily heavy.

  Instead of adding a set~$I$ of independent edges to~$G$, we add a
  star~$S_t$ with $t=\binom{m}{2}+1$ edges, where $m=|E(G)|$.  Then,
  $\mrcr(G)<t$.  The advantage of the star is that all its edges are
  incident to the same vertex and, hence, cannot cross each other.
  Let~$G'=G+S_t$ be the resulting graph.  Exactly as for the set~$I$
  above, we argue that all edges of~$S_t$ must be crossed by the same
  number of edges of~$G$, and must in fact form a cut of~$G$.  Hence,
  we get
  \[t \cdot \mcut(G) \le \mrcr(G') \le t\cdot\mcut(G)+\mrcr(G) < t
  \cdot (\mcut(G)+1).\] This yields $\mcut(G) = \lfloor \mrcr(G')/t
  \rfloor$.  Hence, any $\alpha$-approximation for maximum rectilinear
  crossing number yields an $\alpha$-approximation for
  \textsc{MaxCut}.
\end{proof}

With the same argument, we also obtain hardness of approximation for
$\mccr$, which was only shown NP-hard by Bald et al.~\cite{Bald16}.
The reason is that in the convex setting, too, the ``heavy obstacle''
splits the vertex set into a ``left'' and a ``right'' side.

\begin{corollary}
  Given a graph~$G$, $\mccr(G)$ cannot be approximated better than
  \textsc{MaxCut}.
\end{corollary}

Next we consider the \emph{weighted} topological case, which is
formally is defined as follows.  For a graph~$G$ with positive edge
weights $w \colon E \to \mathbb{Q}_{>0}$ and a drawing~$D$ of~$G$, let
$\mwtcr(D)=\sum_{e \text{ crosses } e'} w(e) \cdot w(e')$ be the
weighted maximum crossing number of~$D$, and let $\mwtcr(G)$ be the
maximum $\mwtcr(D)$ over all drawings $D$ of $G$ be the weighted
maximum crossing number of~$G$.  Let \textsc{MaxWtCrNmb} be the
problem of computing the weighted maximum crossing number of a given
graph.

Compared to the rectilinear and the convex case above, the difficulty
of the topological case is that an obstacle (such as the heavy star
above) does not necessarily separate the vertices into ``left'' and a
``right'' groups any more.  Instead, our new obstacle separates the
vertices into an ``inner'' group and an ``outer'' group, which allows
us to reduce from a cut-based problem.

Our new starting point is the NP-hard problem
\textsc{3MaxCut}~\cite{y-nednp-STOC78}, which is the special case of
\textsc{MaxCut} where the input graph is required to be 3-regular.

\begin{theorem}
  \label{thm:weighted-topological}
  Given an edge-weighted graph~$G$, computing $\mwtcr(G)$ is
  NP-complete.
\end{theorem}

\begin{proof}
  Clearly, topological crossing maximization is in $\cal NP$ since we
  can guess a rotation system for the given graph and, for each edge,
  the ordered subset of the other edges that cross it.  In polynomial
  time, we can then check whether (a)~the weights of the crossings sum
  up to the given threshold, and (b)~the solution is feasible, simply
  by realizing the crossings via dummy vertices of degree~4 and
  testing for planarity of the so-modified graph.

  To show NP-hardness, we reduce from \textsc{3MaxCut}.  Given an
  instance of \textsc{3MaxCut}, that is, a 3-regular graph~$G$, we
  construct an instance of topological crossing maximization, that is,
  a weighted graph~$G'$.  Let~$G'$ be the disjoint union of~$G$ with
  edges of weight~1 and a single triangle~$T$ with edges of (large)
  weight~$t$.  Any edge of~$G$ that connects a vertex in the interior
  of~$T$ to a vertex in the exterior of~$T$ can cross~$T$ up to three
  times (that is, each edge of~$T$ once).  Any edge that connects two
  vertices in the interior (or two vertices in the exterior) of~$T$
  can cross~$T$ at most twice.  In any 3-regular graph $(V,E)$, it
  holds that $|E| = 3/2\cdot|V|$.  Due to the 3-regularity of~$G$, we
  have that, for each vertex in~$V$, at least two of its incident
  edges are in a maximum cut.  Hence, $\mcut(G) \ge 2/3\cdot m = n$,
  where $n$ and $m$ are the numbers of the vertices and edges of~$G$.
  Let~$C=(V_1,V_2)$ be any maximum cut of~$G$.  Since any vertex has
  at most one edge that does not cross~$C$, the edges in $G[V_1]$ form
  a matching~$M_1$ and the edges in~$G[V_2]$ form a matching~$M_2$.

  \begin{figure}[tb]
    \begin{minipage}[b]{.59\textwidth}
      \centering
      \includegraphics{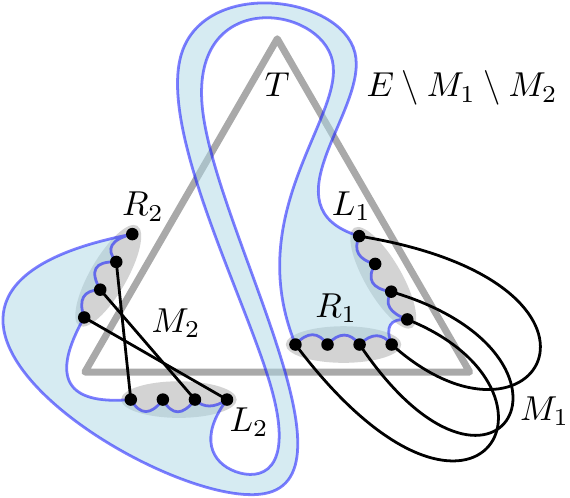}
    \end{minipage}
    \hfill
    \begin{minipage}[b]{.37\textwidth}
      \centering
      \includegraphics{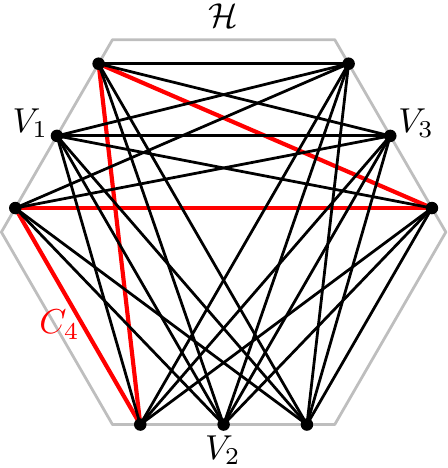}
    \end{minipage}

    \begin{minipage}[t]{.59\textwidth}
      \caption{Given a 3-regular graph~$G$, a drawing of~$G'=G+T$ with
        the maximum number of crossings yields a maximum cut of~$G$ if
        the edges of the triangle~$T$ have much larger weight than the
        edges of~$G$.  The edges (in the light blue region) that
        cross~$T$ trice are in the cut.}
      \label{fig:triangle}
    \end{minipage}
    \hfill
    \begin{minipage}[t]{.37\textwidth}
      \caption{A cross\-ing-maximal drawing of the complete tripartite
        graph $K_{k,k,k}$.}
      \label{fig:hexagon}
    \end{minipage}
  \end{figure}

  Consider a drawing of~$G'$ as in Fig.~\ref{fig:triangle}.  For $i
  \in \{1,2\}$, partition the vertices in~$V_i$ into a left
  subset~$L_i$ and a right subset~$R_i$ so that all edges in~$M_i$ go
  from left to right.  Each edge in the cut crosses all edges of~$T$.
  Each edge in~$M_1 \cup M_2$ crosses exactly two edges of~$T$.
  Clearly, $\mcr(G) \le \binom{m}{2} = \binom{3n/2}{2} < 9/8\cdot
  n^2$.  To ensure that one crossing of an edge of~$T$ contributes
  more than this, we set $t = 9/8\cdot n^2$.  Since any edge of~$G$
  crosses triangle~$T$ at least twice, we get the lower bound
  $\mwtcr(G') \ge t(2m+\mcut(G))$ and the upper bound $\mwtcr(G') \le
  2mt+t\cdot\mcut(G)+\mcr(G) < t (2m+\mcut(G)+1)$, which yields
  $\mcut(G) = \lfloor \mwtcr(G')/t \rfloor - 2m$.
\end{proof}
In Appendix~\ref{sec:ptas} we argue why it is unlikely that
\textsc{MaxWtCrNmb} admits a PTAS.

We now set out to strengthen the result of
Theorem~\ref{thm:weighted-topological}; we want to show that even the
unweighted maximum crossing number is hard to compute.  Observe that
in the above proof, the given graph~$G$ from the \textsc{3MaxCut}
instance remained unweighted, but we required a heavily weighted
additional triangle~$T$.  Our goal is now, essentially, to substitute
$T$ with an unweighted structure that serves the same purpose.
Unfortunately, due to the large number of crossings of this new
structure, we cannot make any statement about non-approximability of
the unweighted case.  The na\"ive approach of simply adding multiple
unweighted triangles does not easily work since already the
entanglement of the triangles among each other is non-trivial to
argue.

\begin{theorem}
  Given a graph~$G$, $\mcr(G)$ is NP-complete to compute.
\end{theorem}

\begin{proof}
  The membership in $\cal NP$ follows from
  Theorem~\ref{thm:weighted-topological}.  To argue hardness, given an
  instance~$G$ of \textsc{3MaxCut}, we construct an unweighted
  graph~$G'$---the instance for computing $\mcr(G')$---as the disjoint
  union of $G$ and a complete tripartite graph $K:=K_{k,k,k}$ with $k$
  vertices per partition set, $k>\sqrt{9/8}\cdot n$.  A result of
  Harborth~\cite{harborth1976parity} yields $\mcr(K) = \binom{3k}{4} -
  3\binom{k}{4} - 6k\binom{k}{3} \in \Theta(k^4)$.

  We first analyze a crossing-maximal drawing of~$K$; see
  Fig.~\ref{fig:hexagon}.  Consider a straight-line drawing ``on a
  regular hexagon $\mathcal{H}$''.  Let $V_1,V_2,V_3$ be the partition
  sets of~$K$ and label the edges of~$\mathcal{H}$ cyclically
  $1,2,\ldots,6$.  Place $V_i$, $1\leq i\leq 3$, along edge~$2i$
  of~$\mathcal{H}$.  We claim that $\mcr(K)$ is achieved by this
  drawing. In fact, the arguments are analogous to the maximality of
  the na\"ive drawing for complete bipartite graphs on two layers: a
  4-cycle can have at most one crossing.  In the above drawing, every
  4-cycle has a crossing.  On the other hand, any crossing in
  \emph{any} drawing of~$K$ is contained in a 4-cycle.

  Intuitively, when thinking about shrinking the sides $1,3,5$
  in~$\mathcal{H}$, we obtain a drawing akin to $T$ in the hardness
  proof for the weighted maximum crossing number.  It remains to argue
  that there is an optimal drawing of full~$G'$ where $K$ is drawn as
  described.  Consider a drawing realizing $\mcr(G')$ and note that
  any triangle in $K$ is formed by a vertex triple, with a vertex from
  each partition set.  Pick a triple $\tau=(v_1,v_2,v_3)\in V_1 \times
  V_2 \times V_3$ that induces a triangle~$T_\tau$ with maximum number
  of crossings with~$G$ among all such triangles.  Now, redraw~$K$
  along~$T_\tau$ according to the above drawing scheme such that, for
  $i=1,2,3$, it holds that (a)~all vertices of~$V_i$ are in a small
  neighborhood of~$v_i$ and (b)~any edge $(w_i,w_j) \in V_i \times
  V_j$ for some $j \ne i$ crosses exactly the same edges of~$G$ as the
  edge $(v_i,v_j)$.  Our new drawing retains the same crossings
  within~$G'$, achieves the maximum number of crossings within~$K$,
  and does not decrease the number of crossings between~$K$ and~$G$;
  hence it is optimal.  In this drawing, $K$ plays the role of the
  heavy triangle~$T$ in the hardness proof of the weighted case, again
  yielding NP-hardness.
\end{proof}

\section{Conclusions and Open Problems}

We have considered the crossing maximization problem in the
topological, rectilinear, and convex settings.  In particular, we
disproved a conjecture of Alpert et al.~\cite{alpert2009} that the
maximum crossing number in the latter two settings always coincide.
On the other hand, we propose a new setting, the ``separated drawing''
setting, and ask whether for every bipartite graph the maximum
rectilinear, maximum convex, maximum separated, and maximum separated
convex crossing numbers coincide.

Concerning complexity, we have shown that the maximum rectilinear
crossing number is APX-hard and the maximum topological crossing
number is NP-hard.  A natural question then is whether the maximum
topological crossing number is also APX-hard.  We have shown this to
be true in the weighted topological case.  It also remains open
whether rectilinear crossing maximization is in $\cal NP$, which would
have followed if the rectilinear and convex setting were equivalent as
conjectured by Alpert et al..
A reviewer of an earlier version of this paper was wondering about
the complexity of maximum crossing number for planar graphs.
For planar graphs, \textsc{MaxCut} is tractable and our hardness
arguments no longer apply, leaving open the question of the complexity
of computing the maximum crossing number for this graph class.

Other intriguing crossing maximization problems remain open: apart
from the two classic problems that we mentioned above---Conway's
Thrackle Conjecture and Ringeisen's Subgraph Problem---we are
interested in the separation of the rectilinear and the separated
convex setting for bipartite graphs.


\paragraph*{Acknowledgments.}

Work on this problem started at the 2016 Bertinoro Workshop of Graph
Drawing.  We thank the organizers and other participants for
discussions, in particular Michael Kaufmann.  We also thank Marcus
Schaefer, G\'abor Tardos, and Manfred Scheucher.

\bibliographystyle{alphaurl} 
\bibliography{abbrv,crossings}

\bigskip 
\appendix

\noindent\textbf{\Large Appendix}

\section{Counterexamples with 12 Vertices}
\label{sec:small}

Here we provide three similar graphs with $12$ vertices and $16$ edges
violating the convexity conjecture (Conjecture~\ref{c-general}).  Note
that each graph is planar and has maximum degree~$4$ or~$5$.  This
shows that the convexity conjecture is false also for some natural
graph classes such as planar graphs or graphs with maximum degree at
most four.  Our proof is based on a relatively long
case-analysis. Manfred Scheucher independently verified by a computer
search that these three graphs indeed violate the convexity
conjecture. Moreover, his unsuccessful attempts to find a smaller
counterexample with the use of computer search support our feeling
that the convexity conjecture might hold for all graphs on at most
$11$ vertices.

Let $H$ be the graph with $10$ vertices and $12$ edges from the
previous subsections. We distinguish three types of vertices:
$A$-vertices, $B$-vertices, and $C$-vertices. The central vertex is
the only {\em $A$-vertex}.  The three vertices $v_0,v_3,v_6$ of $H$
connected to the central vertex are the {\em $B$-vertices} and the six
vertices in $H$ of degree two are {\em $C$-vertices}.  The three edges
adjacent to the $A$-vertex are called $\alpha$-edges, the six edges
connecting a $B$-vertex with a $C$-vertex are called $\beta$-edges and
the remaining three edges connecting independent pairs of $C$-vertices
are called $\gamma$-edges. The nine $B$- and $C$-vertices are {\em
  cycle vertices}, and the nine $\beta$- and $\gamma$-edges forming a
$9$-cycle are called {\em cycle edges}.

\begin{figure}[hb]
  \centering
  \includegraphics{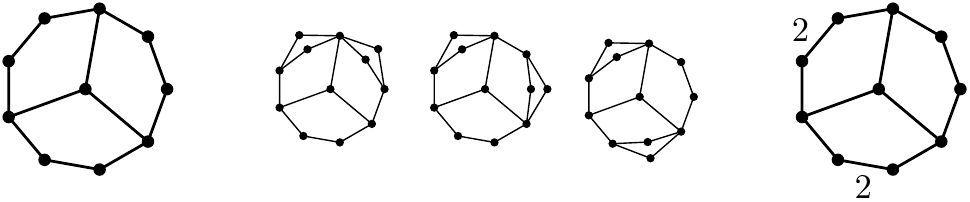}%
  \caption{The graph $H$ (left), the three possible graphs $H_{16}$
    (middle), and the weighted graph $W_{14}$ (right).}
  \label{fig:small-counterexamples}
\end{figure}


We choose $H_{16}$ as a graph obtained from $H$ by selecting a pair of
non-adjacent $C$-vertices and replacing each of them by a pair of
independent vertices. Since the $C$-vertices have degree two in $H$,
four edges of $H$ are replaced by a copy of $K_{1,2}$, thus the graph
$H_{16}$ has $12$ vertices and $16$ edges. Up to isomorphism, $H_{16}$
is one of the three planar graphs depicted in the middle of
Fig.~\ref{fig:small-counterexamples}.
It corresponds to the weighted graph $W_{16}$ which is the graph $H$
with edge weights, where two of the $\beta$-edges have weight two, two
of the $\gamma$-edges have weight two, and the remaining eight edges
have weight one. Further, let $W_{14}$ be the same weighted graph with
the exception that all the $\beta$-edges have weight one, see the
right of Fig.~\ref{fig:small-counterexamples}.  Thus, only two
$\gamma$-edges have weight two, otherwise the edges in $W_{14}$ have
weight one. The graph $W_{14}$ is, up to isomorphism, uniquely
determined regardless of the graph $H_{16}$.

We now give two lemmas used in the proof that $H_{16}$ is a
counterexample for the convexity conjecture.

\begin{lemma}\label{l:cycle-edge-avoids}
  In any drawing of $H$, any cycle edge avoids another edge.
\end{lemma}

\begin{proof}
  Let $e$ be a cycle edge. Then there is a $5$-cycle $Z$ consisting of
  edges non-adjacent to $e$.  (The cycle $Z$ contains two
  $\alpha$-edges, two $\beta$-edges and one $\gamma$-edge.) There must
  be two consecutive vertices of $Z$ lying on the same side of the
  edge $e$ in the considered drawing.  The edge connecting these two
  vertices is avoided by $e$.
\end{proof}

\begin{lemma}\label{l:span-vs.-avoidances}
  In any convex drawing of $H$, any cycle edge of span $s\in\{0,1,2\}$
  avoids at least $6-2s$ cycle edges.
\end{lemma}

\begin{proof}
  Let $e$ be a cycle edge of span $s$. We first give an upper bound on
  the number of edges incident to $e$.  The edge $e$ is incident to
  exactly one cycle edge at each of its two vertices.  Since every
  cycle edge intersecting $e$ is incident to one of the $s$ cycle
  vertices of the ``span interval'' of $e$, at most $2s$ cycle edges
  intersect $e$. Altogether, at most $2+2s$ cycle edges different from
  $e$ have a point in common with $e$. Since there are eight cycle
  edges different from $e$, the edge $e$ avoids at least
  $8-(2+2s)=6-2s$ cycle edges.
\end{proof}

We now fix a convex drawing $D$ of $H_{16}$ maximizing the number of
crossings and with twins placed next to each other. It gives a convex
drawing of the weighted graph $W_{14}$ in the way described above.
Since there is a non-convex drawing of $H_{16}$ with loss $13$, we
need to show that the loss of the drawing $D$ of $H_{16}$ is at least
$14$.  From Lemma~\ref{l:cycle-edge-avoids}, applied on the drawing
$D$, the loss of $H_{16}$ and the weighted loss of the corresponding
drawing of $W_{14}$ differ by at least two.  Thus, it suffices to show
that the weighted loss of the drawing of $W_{14}$ given by the drawing
$D$ is at least $12$. Before proving it, we fix some notation.

The nine $B$- and $C$-vertices of $W_{14}$ are denoted by
$1,2,\dots,9$ in the counterclockwise order in which they appear in
the drawing $D$.  Without loss of generality we may assume that the
$B$-vertices are $1$, $j$, $k$, where $1<j<k\le9$ and the vertex $A$
lies in the counterclockwise interval $(k,1)$. In other words, the
three vertices $k,A,1$ appear in this counterclockwise order along the
convex hull of the vertex set of $D$.

In the following, if a $\beta$-edge avoids a $\gamma$-edge in $D$, we
say that there is a {\em $\beta\gamma$-avoidance}.  Similarly we
define {\em $\beta\beta$-avoidances} as avoidances of pairs of the
$\beta$-edges, and {\em $\gamma\gamma$-avoidances} as avoidances of
pairs of the $\gamma$-edges.  Finally, {\em $\alpha*$-avoidances} are
avoidances of pairs of edges that contain an $\alpha$-edge.

\begin{lemma}\label{l:k-vs.-avoidances}
  There are at least $2(k-2)$ $\alpha*$-avoidances.
\end{lemma}

\begin{proof}
  Let $X$ be the set of the $k-2$ vertices $2,3,\dots,k-1$.  If a
  cycle edge connects two vertices of $X$ then it avoids the
  $\alpha$-edges $A1$ and $Ak$.  If a cycle edge is incident to one of
  the vertices of $X$ then it avoids one of the $\alpha$-edges $A1$
  and $Ak$.  Thus, for each cycle edge $e$, the number of
  $\alpha$-edges avoided by $e$ is at least as big as the number of
  incidences of $e$ with $X$. Since the total number of incidences of
  the vertices in $X$ with the cycle edges is exactly $2|X|=2(k-2)$,
  the number of $\alpha*$-avoidances is at least $2(k-2)$.
\end{proof}

We now distinguish six cases.

\paragraph{Case 1: $k=3$ and there is no $\gamma\gamma$-avoidance.}
In this case the $\beta$-vertices are $1$, $2$, $3$ and the three
$\gamma$-edges are $47$, $58$, and $69$. Each of them has span~$2$ and
therefore, by Lemma~\ref{l:span-vs.-avoidances}, it avoids at least
two of the $\beta$-edges.  Since the total weight of the
$\gamma$-edges is~$5$, the $\beta\gamma$-avoidances have total weight
at least~$10$.  Since there are at least two $\alpha*$-avoidances by
Lemma~\ref{l:k-vs.-avoidances}, we get that the weighted loss of the
drawing of $W_{14}$ (i.e., the total weighted number of avoidances) is
at least $12$ in Case~1.

\paragraph{Case 2: $k=3$ and there is a $\gamma\gamma$-avoidance.}
The $\beta$-edge $\beta_4$ containing the vertex $4$ has the five
$C$-vertices $5$, $6$, $7$, $8$, $9$ on the same side and therefore
avoids two $\gamma$-edges. Since any two $\gamma$-edges have total
weight three or four, it follows that $\beta_4$ appears in
$\beta\gamma$-avoidances of total weight at least three. By symmetry,
$\beta_9$ also appears in $\beta\gamma$-avoidances of total weight at
least three.

The edge $\beta_5$ has the four $C$-vertices $6$, $7$, $8$, $9$ on the
same side and therefore avoids at least one $\gamma$-edge.  By
symmetry, $\beta_8$ also avoids at least one $\gamma$-edge.

Summarizing, the edges $\beta_5,\beta_6,\beta_8,\beta_9$ appear in
$\beta\gamma$-avoidances of total weight at least $3+1+1+3=8$.
Additionally, there are two $\alpha*$-avoidances and there is a
$\gamma\gamma$-avoidance which is necessarily of weight two or
four. It follows that the avoidances have total weight at least
$8+2+2=12$.

\paragraph{Case 3: $k=4$ and there is no $\gamma\gamma$-avoidance.}
Without loss of generality, we assume that the $B$-vertices are
$1,3,4$. Then the $\gamma$-edges are $27$, $58$, $69$. The edge $58$
avoids the $\beta$-edges $\beta_2$ and $\beta_9$.  Similarly, the edge
$69$ avoids the $\beta$-edges $\beta_2$ and $\beta_5$.  Since the
edges $58$ and $69$ have total weight three or four, they appear in
$\beta\gamma$-avoidances of total weight at least $3\cdot 2=6$.

The edge $\beta_2$ avoids either the two $\beta$-edges incident to the
$C$-vertex $1$ or the two $\beta$-edges incident to the $C$-vertex
$4$. Thus, there are at least two $\beta\beta$-avoidances.  Also,
there are at least four $\alpha*$-avoidances by
Lemma~\ref{l:k-vs.-avoidances}.  Altogether, the avoidances have total
weight at least $6+2+4=12$.

\paragraph{Case 4: $k=4$ and there is a $\gamma\gamma$-avoidance.}
As in Case 3, we assume that the $B$-vertices are $1,3,4$.  The edge
$\beta_2$ avoids two of the three $\gamma$-edges, which gives two
$\beta\gamma$-avoidances of total weight three or four. The edge
$\beta_2$ also avoids at least one $\beta$-edge connecting one of the
vertices $1$ and $4$ with one of the vertices in the interval $[5,9]$.

Since there is a $\gamma\gamma$-avoidance, the interval $[5,9]$
contains the vertices of a $\gamma$-edge $\gamma_0$ of span at most
$1$. The edge $\gamma_0$ avoids at least one $\gamma$-edge and at
least two $\beta$-edges different from $\beta_2$ (for example, if
$\gamma_0$ connects vertices $6$ and $8$, it avoids the $\beta$-edges
$\beta_5$ and $\beta_9$).  The $\gamma\gamma$-avoidance has weight two
or four, and the two $\beta\gamma$-avoidances have total weight at
least two.

Summarizing, avoidances involving no $\alpha$-edge have total weight
at least $3+1+2+2=8$.  Since there are at least four
$\alpha*$-avoidances by Lemma~\ref{l:k-vs.-avoidances}, all avoidances
have total weight at least $8+4=12$.

\paragraph{Case 5: $k=5$.}
The two $\beta$-edges with both vertices in the interval $[1,5]$ have
span at most $2$, and therefore appear in at least four avoidances
among cycle edges.  There are at least six $\alpha*$-avoidances by
Lemma~\ref{l:k-vs.-avoidances}.  It follows that there are at least
ten avoidances.

Since each of the two $\gamma$-edges of weight two avoids another
edge, there are at least two avoidances of weight two or an avoidance
of weight four. We conclude that all the avoidances have total weight
at least $10+2=12$.

\paragraph{Case 6: $k\geq6$.}
Suppose first that all nine cycle edges have span three. Then the
cycle edges form the cycle $162738495$, The $B$-vertices are $1$, $4$,
$7$, the $\gamma$-edge $26$ avoids the two $\alpha$-edges $A1$ and
$A7$, and each of the other eight cycle edges avoids exactly one of
the $\alpha$-edges $A1$, $A4$, $A7$.  Thus, there are ten
$\alpha*$-avoidances. Since each of the two $\gamma$-edges of weight
two appears in at least two avoidances, the total weight of avoidances
is at least $10+2=12$.

Suppose now that there is a cycle edge with span smaller than three.
Then this edge avoids at least two cycle edges. Additionally there are
at least eight $\alpha*$-avoidances.  Altogether there are at least
$2+8=10$ avoidances. Since each of the two $\gamma$-edges of weight
two appears in some avoidance, all the avoidances have total weight at
least $10+2=12$.

\section{It is Unlikely that \textsc{MaxWtCrNmb} Admits a PTAS}
\label{sec:ptas}

Due to the additive term $2m$ in the lower and upper bound for
$\mwtcr(G)$ (see the sequence of inequalities at the end of the proof
of Theorem~\ref{thm:weighted-topological}), the existence of a PTAS
for
\textsc{MaxWtCrNmb} does not directly imply a PTAS for
\textsc{3MaxCut}.  A PTAS for \textsc{MaxWtCrNmb} would, however, give
us a very good estimation of the quantity $q=2m+\mcut(G)$.  Since $G$
is 3-regular, we know that $2m/3 \le \mcut(G) \le m$.  Hence, assuming
a $(1-\varepsilon)$-approximation of $\mcut(G)$, the ratio between the
smallest and the largest possible value of~$q$ is
$(8m/3-\varepsilon)/(3m)=8/9-\varepsilon'=0.\overline{8}-\varepsilon'$.
This would be the approximation ratio of an algorithm for
\textsc{3MaxCut} based on a hypothetical PTAS for \textsc{MaxWtCrNmb}.
\textsc{3MaxCut} is APX-hard; the best known inapproximability ratio
is 0.997 \cite{bk-stir-ICALP99}, which is too large to yield a
contradiction to the existence of a PTAS for \textsc{MaxWtCrNmb}.
However, to the best of our knowledge, the best approximation
algorithm for \textsc{3MaxCut} is the semidefinite program of Goemans
and Williamson \cite{gw-iaamc-JACM95} for general \textsc{MaxCut}.
Its approximation ratio is $\approx 0.878$, and any improvement beyond
this factor, even for the special case of 3-regular graphs, would be
rather unexpected.

\end{document}